\documentclass[12pt]{article}
\usepackage{amsmath}
\usepackage{graphicx,psfrag,epsf}
\usepackage{enumerate}
\usepackage{url} 

\newcommand{\blind}{0}

\usepackage{amsthm,amsfonts}
\usepackage{comment}
\usepackage[bf]{caption}
\usepackage{tikz}
\usetikzlibrary{positioning,shapes.geometric}
\usepackage{lscape}
\usepackage[backend=bibtex, style=nature, citestyle=authoryear]{biblatex}
\bibliography{bibliography}

\def\ind{\perp\!\!\!\perp}

\newcommand{\var}{\mathbb{V}\text{ar}}
\newcommand{\avar}{\text{a}\mathbb{V}\text{ar}}
\newcommand{\cov}{\mathbb{C}\text{ov}}

\newcommand{\Pb}{\mathbb{P}}
\newcommand{\Pn}{\mathbb{P}_n}
\newcommand{\E}{\mathbb{E}}

\newcommand{\bX}{\mathbf{W}}

\newcommand{\bO}{\mathbf{O}}

\DeclareSymbolFont{bbold}{U}{bbold}{m}{n}
\DeclareSymbolFontAlphabet{\mathbbold}{bbold}
\newcommand{\one}{\mathbbold{1}}
\newtheorem{theorem}{Theorem}
\newtheorem{lemma}{Lemma}

\theoremstyle{remark}
\newtheorem{assumption}{Assumption}

\newtheorem{remark}{Remark}

\begin{document}

\def\spacingset#1{\renewcommand{\baselinestretch}%
{#1}\small\normalsize} \spacingset{1}


\if0\blind
{
  \title{ \bf Paradoxes in instrumental variable studies \mbox{with missing data and one-sided noncompliance} }
  \author{ \\ Edward H. Kennedy
    \thanks{Edward H. Kennedy is Assistant Professor in the Department of Statistics, Carnegie Mellon University, Pittsburgh, PA 15213 (e-mail: edward@stat.cmu.edu), and Dylan S. Small is Professor in the Department of Statistics, The Wharton School, University of Pennsylvania. }\hspace{.2cm}
    \\
    Department of Statistics, Carnegie Mellon University \\ \\
    Dylan S. Small \\ Department of Statistics, The Wharton School, University of Pennsylvania \\ \\ 
    }
  \maketitle
  \setcounter{page}{0}
  \thispagestyle{empty}
} \fi

\if1\blind
{
  \vspace*{.8in}
  \begin{center}
    {\LARGE\bf  }
\end{center}
  \setcounter{page}{0}
  \medskip
} \fi

\begin{abstract}
It is common in instrumental variable studies for instrument values to be missing, for example when the instrument is a genetic test in Mendelian randomization studies. In this paper we discuss two apparent paradoxes that arise in so-called single consent designs where there is one-sided noncompliance, i.e., where unencouraged units cannot access treatment. The first paradox is that, even under a missing completely at random assumption, a complete-case analysis is biased when knowledge of one-sided noncompliance is taken into account; this is not the case when such information is disregarded. This occurs because incorporating information about one-sided noncompliance induces a dependence between the missingness and treatment. The second paradox is that, although incorporating such information does not lead to efficiency gains without missing data, the story is different when instrument values are missing: there, incorporating such information changes the efficiency bound, allowing possible efficiency gains. This is because some of the missing values can be filled in, based on the fact that anyone who received treatment must have been encouraged by the instrument (since the unencouraged cannot access treatment).
\end{abstract}

\noindent%
\vfill

\thispagestyle{empty}

\newpage

\spacingset{1}

\section{Introduction}

Instrumental variable methods are a popular approach to causal inference in settings where unmeasured variables confound the relationship between the treatment and outcome of interest. In general, such unmeasured confounding precludes identification of causal effects, and one must resort to  bounds and/or sensitivity analysis. However, in the presence of an instrument some progress can still be made. An instrument is a special variable that affects receipt of treatment but does not directly affect outcomes, and is itself unconfounded. This setup can be represented graphically as in Figure \ref{fig:dag}; formal identifying assumptions are given in the next section. 

\begin{figure}[h!]
\begin{center}
\begin{tikzpicture}[->, shorten >=2pt,>=stealth, node distance=1cm, noname/.style={ ellipse, minimum width=5em, minimum height=3em, draw } ]
\node[] (1) {$Z$};
\node (2) [right= of 1] {$A$};
\node (3) [right= of 2] {$Y$};
\node (5) [below= of 2] {$U$};
\path (1) edge node {} (2);
\path (1) edge [bend left=50pt, dashed,lightgray] node {} (3);
\path (2) edge node {} (3);
\path (5) edge [dashed,lightgray] node {} (1);
\path (5) edge node {} (3);
\path (5) edge node {} (2);
\end{tikzpicture}
\end{center}
\caption{Directed acyclic graph showing instrument $Z$, treatment $A$, outcome $Y$, and unmeasured variables $U$. Gray dotted arrows indicate relationships that are assumed absent by identifying assumptions. \label{fig:dag}}
\end{figure}
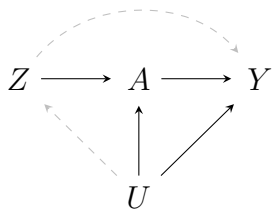

A classic example of an instrument occurs in randomized trials with noncompliance, in which case the assigned treatment is often a reasonable instrument for effects of the treatment that was actually received. Instrumental variable methods have been around for nearly a century \autocite{wright1928tariff, wright1934method}, but their placement in a formal potential outcomes framework occurred only relatively recently \autocite{angrist1996identification}. Examples abound in the literature, including instruments based on distance, treatment provider preference, calendar time, and genetic variants; we refer to  \textcite{hernan2006instruments, baiocchi2014instrumental} for overviews.

Although there is an extensive literature on instrumental variables, the treatment of missing data in such settings is scant. Apart from a few recent exceptions \autocite{burgess2011missing, mogstad2012instrumental, chaudhuri2016gmm, kennedy2018efficient} this problem has not received much attention, despite the fact that missing data is very common in instrumental variable studies. For example, Mendelian randomization studies use genetic variants as instruments, but this information is frequently missing due to subjects not sending in samples, or ambiguous output from genotyping platforms. \textcite{burgess2011missing} reported missingness in SNP-based instruments ranging between 2\% and 11\%; more examples can be found in \textcite{mogstad2012instrumental, chaudhuri2016gmm}. In this paper we consider instrumental variable studies with instrument missingness, and also where noncompliance is one-sided, i.e., where unencouraged units cannot access treatment. One-sided noncompliance common in practice, occurring for example in studies of new drugs not yet on market, and of limited-access job training programs. Although one-sided noncompliance is often associated with experiments, it also occurs in observational studies where missing instrument values are common \autocite{frolich2013identification, kennedy2018efficient}. 

For example, missing instrument data and one-sided noncompliance often arise together in fuzzy regression discontinuity designs; these are instrumental variable analyses where the instrument is an indicator for being above a treatment-influencing threshold \autocite{hahn2001identification, imbens2008regression}. For instance, \textcite{pitt1998impact} and \textcite{frolich2013identification} discuss a study of the effects of microcredit programs in Bangladesh, where the instrument was an indicator for owning more than half an acre of land; those with more than half an acre were ineligible for the program, while those with less self-selected in. \textcite{pitt1998impact} reported a ``substantial number'' of missing instrument values. \textcite{battistin2008ineligibles} give several similar examples, including studies where the instrument is an indicator for being above a test score cutoff. Students below the cutoff may be restricted from enrollment (e.g., in college programs), whereas students above get to choose whether to participate; and test scores are often missing. \textcite{angrist2015wanna} used such an instrument to study the effect of education at Boston's selective ``exam schools''; they mentioned excluding any subjects with missing test scores.

In this paper we discuss two paradoxes that arise in these kinds of instrumental variable studies with missing instrument values and one-sided noncompliance.  The first paradox is that, even under a missing completely at random assumption, a complete-case analysis is biased when knowledge of one-sided noncompliance is taken into account. Surprisingly, this is not the case when such information is disregarded: so discarding information avoids bias. The second paradox is that, although incorporating information about one-sided noncompliance does not lead to efficiency gains without missing data, the story is different when instrument values are missing: there, incorporating such information changes the efficiency bound, allowing possible efficiency gains. Before describing these paradoxes we first present some general efficiency theory for instrumental variable studies with missing instrument values.

\section{Setup \& notation}

Suppose the full data would consist of a sample of observations $\bO^* = (Z,A,Y)$ with $Z$ a binary instrument (e.g., a randomization indicator in experimental settings), $A$ a binary treatment, and $Y$ some real-valued outcome. Starting in Section 4, we consider studies with one-sided noncompliance, in which $Z=0$ implies $A=0$, i.e., control subjects cannot access treatment. Such studies are sometimes referred to as ``single consent designs'' \autocite{zelen1979new}.

Unfortunately, however, we do not observe the full data; instead we observe a sample of independent and identically distributed observations $(\bO_1, ..., \bO_n)$, where
$$ \bO=(RZ,R,A,Y) \sim \Pb $$ 
and $R$ is an indicator of whether the instrument $Z$ is observed or not. When $R=1$ the instrument $Z$ is observed, but when $R=0$ the instrument $Z$ is missing and we only see $RZ=0$ regardless of the value of $Z$. To focus ideas and simplify notation, we do not consider baseline covariates; however extensions are mostly straightforward  (alternatively, all results can be viewed as implicitly conditional on any such covariates).

Our goal is to estimate the classical instrumental variable estimand 
\begin{equation} \label{eq:psi_full}
\Psi^f  \equiv \Psi^f(\Pb) = \frac{\E(Y \mid Z=1) - \E(Y \mid Z=0)}{\E(A \mid Z=1) - \E(A \mid Z=0)} , 
\end{equation}
which equals a treatment effect under standard causal assumptions. (Note that expectations $\E=\E_\Pb$ are under the true $\Pb$, and likewise for $\Psi^f=\Psi^f(\Pb)$ unless stated otherwise). These assumptions have been detailed at length elsewhere \autocite{angrist1996identification, hernan2006instruments}, so we will limit our discussion of them before focusing on the observed data parameter $\Psi^f$ above. We also note that all subsequent results hold for the observed data parameter $\Psi^f$, regardless of whether the causal assumptions are plausible or not.

We let $Y_a$ denote the potential outcome \autocite{rubin1974estimating} that would have been observed had treatment been set to $A=a$; thus the goal is to learn about the distribution of the effect $Y_{a=1}-Y_{a=0}$. We will also need to define potential outcomes under interventions on the instrument; thus let  $Y_{za}$ denote the potential outcome that would have been observed under both $Z=z$ and $A=a$, and similarly let $A_z$ and $Y_z=Y_{zA_z}$ denote the potential treatment and outcome under only $Z=z$. 

To illustrate one set of assumptions under which the observed data parameter $\Psi^f(\Pb)$ represents a causal effect, consider the following.

\begin{assumption}[Consistency] $Y=Y_z$ and $A=A_z$ if $Z=z$, and $Y=Y_{za}$ if $(Z,A)=(z,a)$.
\end{assumption}
\begin{assumption}[Positivity] $0< \Pb(Z=z)< 1$.
\end{assumption}
\begin{assumption}[Instrumentation] $\Pb(A_{z=1} =1) \neq \Pb(A_{z=0}=1)$.
\end{assumption}
\begin{assumption}[Unconfoundedness of $Z$] $Z \ind (A_z, Y_z)$.
\end{assumption}
\begin{assumption}[Exclusion Restriction] $Y_{za}=Y_a$.
\end{assumption}
\begin{assumption}[Monotonicity] $A_{z=1} \geq A_{z=0}$.
\end{assumption}

Consistency means we get to observe potential outcomes (and treatments) under the observed instrument values; this requires there to be no interference, i.e., one subject's treatments and outcomes cannot be affected by other subjects' instrument values or treatments. This would be violated in for example vaccine studies with herd immunity, or studies with strong network structure. Positivity requires everyone to have some chance at each instrument value. Instrumentation means the instrument has to have some effect on treatment; in other words, the arrow from $Z$ to $A$ in Figure \ref{fig:dag} must be present. Unconfoundedness means the instrument must be assigned essentially at random, i.e., the arrow from $U$ to $Z$ in Figure \ref{fig:dag} must be absent. The exclusion restriction says that the instrument can only affect outcomes indirectly through treatment, i.e., the arrow from $Z$ to $Y$ in Figure \ref{fig:dag} must be absent. Monotonicity means there are no ``defiers'' who take treatment when \textit{not} encouraged to by the instrument, but take control when encouraged towards treatment. Note that monotonicity holds by design with one-sided noncompliance, since then $A_{z=0}=0$ so we have $A_{z=1}>A_{z=0}$ if and only if $A_{z=1}=1$. 

The next lemma recalls a result from \textcite{imbens1994identification} showing that under the above assumptions, $\Psi^f$ equals a ``local'' average treatment effect among compliers, i.e., among those who take treatment only when encouraged to do so by the instrument.

\begin{lemma}
Under Assumptions 1--6, the average treatment effect among compliers with $A_{z=1}>A_{z=0}$ is given by
$$ \Psi^f = \E(Y_{a=1} - Y_{a=0} \mid A_{z=1}>A_{z=0}) $$
with $\Psi^f$ defined in \eqref{eq:psi_full}.
\end{lemma}
\begin{proof}
This result follows from \textcite{imbens1994identification, angrist1996identification}. Note
\begin{align}
\E(Y \mid Z=1) - \E(Y \mid Z=0) &= \E(Y_{z=1} - Y_{z=0}) = \E(Y_{a=A_{z=1}} - Y_{a=A_{z=0}}) \nonumber \\
&= \E\{(Y_{a=1} - Y_{a=0}) \one(A_{z=1} > A_{z=0})\} \nonumber \\
&= \E(Y_{a=1} - Y_{a=0} \mid A_{z=1} > A_{z=0}) \Pb(A_{z=1} > A_{z=0}) \label{eq:ivnum}
\end{align}
where the first equality follows from consistency, positivity, and unconfoundedness, the second by consistency and the exclusion restriction, and the third by monotonicity. Now 
\begin{align} \label{eq:ivden}
\Pb(A_{z=1} > A_{z=0})  &= \E(A_{z=1} - A_{z=0}) = \E(A \mid Z=1) - \E(A \mid Z=0) 
\end{align}
where the first equality follows by monotonicity, and the second by consistency, positivity, and unconfoundedness. Finally, to obtain the result we divide \eqref{eq:ivnum} by \eqref{eq:ivden}, which requires the instrumentation assumption.
\end{proof}

We note that the observed data quantity $\Psi^f$ can also represent a causal effect under alternative assumptions. For example, monotonicity is sometimes replaced with effect homogeneity assumptions, in which case $\Psi^f$ can equal the average effect in the entire population \autocite{angrist1996identification} or among the treated \autocite{hernan2006instruments}. 

\section{Preliminaries}

In this section we present some general efficiency theory for instrumental variable studies (not necessarily with one-sided noncompliance) with missing instrument values. In particular we first present the nonparametric efficiency bounds for $\Psi^f$ without any missingness, then go on to present the efficiency bound under missing at random assumptions, describe the corresponding efficiency loss, and then analyze the efficiency of a complete-case estimator under a missing completely at random assumption.

\subsection{Efficiency bound without missingness}

With full data (i.e., when $R=1$ with probability one and $Z$ is always observed), the usual instrumental variable estimator is
\begin{equation} \label{eq:psihat}
 \Psi^f(\Pn) 
= \frac{\Pn(ZY) - \Pn(Z)\Pn(Y)}{\Pn(ZA) - \Pn(Z) \Pn(A) } 
\end{equation}
where $\Pn$ is the empirical measure, so that sample averages can be written using the shorthand $n^{-1} \sum_i f(\bO_i) = \Pn\{f(\bO)\}$. It is straightforward to see that this estimator solves the equation $ \Pn\{ (Z-\widehat\pi_n) (Y-\psi A) \} = 0$ in $\psi$ for $\widehat\pi_n=\Pn(Z)$, so that using standard estimating equation results we have
$$ \Psi^f(\Pn) - \Psi^f = \Pn \{ D(\bO) \} + o_\Pb(1/\sqrt{n})  $$
where 
\begin{equation} \label{eq:eif_full}
D(\bO) = \frac{ \{ Z-\E(Z) \}\{ (Y - \Psi^f A) - \E(Y-\Psi^f A)\} }{ \cov(Z,A)} . 
\end{equation}

The next result gives the efficiency bound in this full data setting.

\begin{lemma}
Suppose the instrument $Z$ is always observed. Then the nonparametric efficiency bound for $\Psi^f$ is given by
$$ \var\left[ \frac{ \{ Z-\E(Z) \}\{ (Y - \Psi^f A) - \E(Y-\Psi^f A)\} }{ \cov(Z,A)} \right] . $$
\end{lemma}
\begin{proof}
Since the full data model is nonparametric, the tangent space is equal to the entire Hilbert space of mean-zero finite-variance functions \autocite{bickel1993efficient, van2003unified, tsiatis2006semiparametric}; therefore $D$ is the only influence function and necessarily the efficient one. This immediately implies that its variance is the efficiency bound.
\end{proof}

\subsection{Efficiency bound under MCAR \& MAR}

Now consider the case where $Z$ can be missing, with $R$ an indicator for whether $Z$ is observed, and assume the ``missing completely at random'' (MCAR) condition
$$ R \ind (Z,A,Y) . $$
This means the missingness in $Z$ is completely unrelated to not only the underlying $Z$ values, but also treatment and outcome. 

The next result gives the efficiency bound in this missing at random setting. In Section 4.1 we use this result to construct an estimator whose asymptotic variance matches the bound under weak conditions. 

\begin{lemma} \label{lem:effbd_mar}
Suppose $R \ind (Z,A,Y)$ so that the instrument $Z$ is missing completely at random. Then the efficiency bound for $\Psi^f$ is given by the variance of 
\begin{equation} \label{eq:eif_mar}
D^*(\bO) = \frac{R}{\E(R \mid A,Y)} \Big[ D(\bO) - \E\{ D(\bO) \mid A, Y, R=1 \}  \Big] + \E\{ D(\bO) \mid A, Y, R=1\} ,
\end{equation}
which can be expressed as 
\begin{equation} \label{eq:bd_mar}
\var\{ D^*(\bO)\} = \frac{\var\{D(\bO)\}}{\E(R)} - \frac{\Pb(R=0)}{\Pb(R=1)} \var\Big[ \E\{ D(\bO) \mid A, Y \}  \Big] . 
\end{equation}
Further, the efficiency bound is the same under the weaker missing at random condition $R \ind Z \mid A,Y$.
\end{lemma}

\begin{proof}
From \textcite{robins1995semiparametric}, efficiency bounds under MCAR (here, $R \ind (Z,A,Y)$) are the same as those under a weaker missing at random (MAR) condition $R \ind Z \mid A,Y$, which is implied by MCAR. Therefore applying their theory (also detailed by \textcite{tsiatis2006semiparametric}) to our setting yields an efficient influence function under both MAR and MCAR given by
$$ D^*(\bO) = \frac{R}{\E(R \mid A,Y)} \Big[ D(\bO) - \E\{ D(\bO) \mid A, Y, R=1 \}  \Big] + \E\{ D(\bO) \mid A, Y, R=1\} . $$
Note that $\E(R \mid A,Y)=\E(R)$ under MCAR, and the conditioning on $R=1$ in the two rightmost terms is unnecessary (under both MCAR and MAR). Hence the efficiency bound under MCAR is given by
\begin{align*}
\var\{ D^*(\bO)\} &= \E\left[\left( \frac{R}{\E(R)} \Big[ D(\bO) - \E\{ D(\bO) \mid A, Y \}  \Big] + \E\{ D(\bO) \mid A, Y \} \right)^2 \right] \nonumber \\
&= \E\left( \frac{R}{\E(R)^2} \Big[ D(\bO) - \E\{ D(\bO) \mid A, Y \}  \Big]^2 + \E\{ D(\bO) \mid A, Y \}^2 \right) \nonumber \\
&= \frac{1}{\E(R)} \E \Big[ \var\{D(\bO) \mid A,Y\} \Big] +  \var\Big[ \E\{ D(\bO) \mid A, Y \}  \Big]  \nonumber  \\
&= \frac{\var\{D(\bO)\}}{\E(R)} - \frac{\Pb(R=0)}{\Pb(R=1)} \var\Big[ \E\{ D(\bO) \mid A, Y \}  \Big] .
\end{align*}
where the third and fourth equalities follow from $R \ind (A,Y)$ and $0=\E\{ D(\bO)\} = \E[\E\{D(\bO) \mid A,Y\}]$. 
\end{proof}

After some rearranging, the result in Lemma \ref{lem:effbd_mar} shows that the relative efficiency under MCAR versus the full data setting is
\begin{align*}
\frac{ \var\{ D^*(\bO)\} }{ \var\{ D(\bO)\}} 
&= 1 +  \text{odds}(R=0) \frac{ \E[ \var\{ D(\bO) \mid A, Y \}]}{\var\{D(\bO)\}} .
\end{align*}
Therefore full data efficiency is only attainable under MCAR in unusual no-variance situations: if $Z$ is constant within strata of $(A,Y)$, or if $(Y-\Psi^f A)$ is constant. 

\subsection{Complete-case efficiency under MCAR}

Under MCAR, a complete-case analysis that simply discards observations with $R=0$ and applies the usual estimator will be valid. Here we describe the efficiency of such an approach.

Specifically, the complete-case version of the instrumental variable estimand presented in the previous subsection is given by
\begin{equation} \label{eq:psi_cc}
\Psi^{cc}(\Pb) =  \frac{\E(Y \mid Z=1,R=1)-\E(Y \mid Z=0,R=1)}{\E(A \mid Z=1,R=1) - \E(A \mid Z=0,R=1)} 
\end{equation}
and we have
$$  \Psi^{f}(\Pb) = \Psi^{cc}(\Pb)$$
by the MCAR condition.  The complete-case estimator is then given by
$$  \Psi^{cc}(\Pn) 
= \frac{ \Pn(RZY) \Pn(R) - \Pn(RZ) \Pn(RY) }{ \Pn(RZA) \Pn(R) - \Pn(RZ) \Pn(RA) } . $$
This estimator solves $\Pn\{ R(Z-\widehat\pi_n) (Y-\psi A) \} = 0 $ in $\psi$ for 
for $\widehat\pi_n=\Pn(Z \mid R=1)=\Pn(RZ)/\Pn(R)$, and thus using standard estimating equation results as before (together with the fact that $R \ind (Z,A,Y)$), it is straightforward to show that
$$ \Psi^{cc}(\Pn) - \Psi^f = \Pn\left\{ \frac{R}{\E(R)} D(\bO)  \right\} + o_\Pb(1/\sqrt{n}) . $$
Therefore since $R \ind (Z,A,Y)$ the asymptotic variance ($\avar$) of the complete-case estimator is
$$ \avar\left\{ \Psi^{cc}(\Pn) \right\} = \var\left\{ \frac{R}{\E(R)} D(\bO)  \right\} =  \frac{\var\{D(\bO)\}}{ \E(R) } =  \frac{\avar\left\{ \Psi^{f}(\Pn) \right\}}{ \E(R) } $$
so that the relative efficiency ${\avar\left\{ \Psi^{cc}(\Pn) \right\}}/{\avar\left\{ \Psi^{f}(\Pn) \right\}}$ is the inverse proportion of units not missing $1/\E(R)=1/\Pb(R=1)$. This is to be expected: since the complete-case estimator only uses units with $R=1$, it requires $1/\Pb(R=1)$ times as many observations to match the efficiency attainable without missingness (for example, one would need a sample size twice as large if half the units have missing data).

Further, comparison with the efficiency bound from Lemma \ref{lem:effbd_mar} shows that the complete-case estimator is generally inefficient, since its asymptotic variance (the first term on the right-hand side of \eqref{eq:bd_mar}) is strictly greater than the efficiency bound (the left-hand side of \eqref{eq:bd_mar}) unless $\E\{D(\bO) \mid A,Y\}$ is constant (hence, the second term on the right-hand side of  \eqref{eq:bd_mar} is zero). 

\section{First paradox}

Now that we have characterized the efficiency theory for instrumental variable studies under the MCAR condition $R \ind (Z,A,Y)$, we will present our first (apparent) paradox that arises when noncompliance is one-sided. Namely, we will show that, even assuming the MCAR condition, a complete-case analysis is generally biased when one-sided noncompliance is taken into account. More specifically, incorporating the fact that noncompliance is one-sided generates a new missingness indicator, for which MCAR (but not MAR) is violated. 

\subsection{Main result}

As mentioned in Section 2, one-sided noncompliance means the unencouraged with $Z=0$ do not have access to treatment, so that $Z=0$ implies $A=0$, i.e., $A_{z=0}=0$ and $A=1$ implies $Z=1$. Equivalently, in such studies there are only never-takers ($A_{z}=0$) and compliers ($A_{z}=z$), and no always-takers ($A_z=1$) or defiers ($A_z=1-z$). In this section we  show that incorporating information about one-sided noncompliance invalidates MCAR so that a complete-case analysis is biased. 

Before stating our first result we introduce some new notation. Since $A=1$ implies $Z=1$ under one-sided noncompliance, whenever $Z$ is missing for subjects with $A=1$ we can actually fill in $Z=1$. This yields a new missingness indicator, different from $R$, defined by 
\begin{equation} \label{eq:rnew}
R^\dagger=R(1-A)+A
\end{equation}
which equals the initial missingness indicator $R$ whenever $A=0$, but equals 1 whenever $A=1$ since then it is known that $Z=1$. We also define the corresponding complete-case estimand based on this updated indicator as
\begin{equation} \label{eq:psi_dag}
\Psi^\dagger(\Pb) = \frac{ \E(Y \mid Z=1,R^\dagger=1) - \E(Y \mid Z=0, R^\dagger=1) }{ \E(A \mid Z=1,R^\dagger=1) - \E(A \mid Z=0, R^\dagger=1) }
\end{equation}

The next theorem gives our first main result, which is that incorporating information about one-sided noncompliance violates MCAR and thus generally invalidates a complete-case analysis, even though without incorporating such information MCAR in fact holds and a complete-case analysis would be valid.

\begin{theorem}
Let $\bO=(R,RZ,A,Y) \sim \Pb$ and suppose the MCAR condition $R \ind (Z,A,Y)$ holds, so that
$$ \Psi^f(\Pb) = \Psi^{cc}(\Pb) $$
for $\Psi^f$ and $\Psi^{cc}$ the full-data and complete-case instrumental variable estimands in \eqref{eq:psi_full} and \eqref{eq:psi_cc}. When there is one-sided noncompliance and $\Pb(R=1)<1$, then MCAR fails for the updated indicator $R^\dagger$ from \eqref{eq:rnew} which incorporates the knowledge that $A=1$ implies $Z=1$, i.e., 
$$ R^\dagger \not\!\ind (Z,A,Y) $$
and in general a complete-case analysis based on $R^\dagger$ will be biased, i.e.,
$$ \Psi^f(\Pb)  \neq \Psi^\dagger(\Pb)  . $$
\end{theorem}

\begin{proof}
In the next subsection we will derive the explicit form of the bias of the complete-case estimand based on $R^\dagger$, proving the second result. To see why $R^\dagger \not\!\ind (Z,A,Y)$ for the first result, note that $R^\dagger \not\!\ind A$ due to the fact that
\begin{align*}
\Pb(R^\dagger=1 \mid A=1) &= 1  \\
\Pb(R^\dagger=1 \mid A=0) &= \Pb(R=1 \mid A=0) = \Pb(R=1) ,
\end{align*} 
and these expressions are not equal unless there is no missingness. 
\end{proof}

The result in Theorem 1 is counterintuitive at first glance: if we use the original missingness indicator $R$ and pretend not to know about one-sided noncompliance, then a complete-case analysis is valid; however if we take into account this latter information using $R^\dagger$, then we cannot do a complete-case analysis anymore. The main intuition for this result is that, even if the instrument is initially missing completely at random (using the indicator $R$), incorporating the one-sided noncompliance information induces a dependence between the new indicator $R^\dagger$ and treatment $A$, since $A=1$ implies $Z=1$ and thus $R^\dagger=1$. This invalidates MCAR and instead makes the instrument missing at random (but not completely).

\begin{remark}
In a setting with baseline covariates $\bX$, the results of Theorem 1 indicate that a $\bX$-specific MAR condition is violated for the indicator $R^\dagger$, i.e., $R^\dagger \not\!\ind (Z,A,Y) \mid \bX$, even when it holds for the original indicator $R$, i.e., even if $R \ind (Z,A,Y) \mid \bX$. 
\end{remark}

\begin{remark}
Although MCAR is violated when using $R^\dagger$, MAR still holds; more specifically we have $R^\dagger \ind (Z,Y) \mid A$. This follows since if $A=1$ then $R^\dagger=1$ is constant, and if $A=0$ then $R^\dagger=R$, which is independent of $(Z,Y)$. 
\end{remark}

Theorem 1 therefore tells us that, in studies with one-sided noncompliance and missing instrument values, one should either not use complete-case analyses (or analogous methods in settings with covariates) or else one should not incorporate the one-sided noncompliance information. In our view it is preferable to assume minimal untestable MAR conditions like $R \ind Z \mid (A,Y)$, since this holds for both $R$ and $R^\dagger$, instead of stronger MCAR conditions like $R \ind (Z,A,Y)$, which incorporate testable constraints ($R \ind (A,Y)$) and require different analysis methods depending on whether one-sided noncompliance information is taken into account. This is especially prudent since efficiency bounds are the same under both MAR and MCAR; thus the benefits of MCAR assumptions do not include efficiency gains, but in fact just less computational burden (e.g., allowing complete-case analysis). 

Thus, we advocate for an estimator that can attain the MCAR/MAR efficiency bound under weak conditions. One such estimator is given by 
\begin{equation} \label{eq:psihat_eff}
 \widehat\Psi_n = \frac{ \Pn( [\frac{R^\dagger}{\widehat\E(R^\dagger \mid A,Y)} \{ Z - \widehat\E(Z \mid A,Y,R^\dagger=1) \} + \widehat\E(Z \mid A,Y,R^\dagger=1)]\{ Y - \Pn(Y) \} ) }{ \Pn( [\frac{R^\dagger}{\widehat\E(R^\dagger \mid A,Y)} \{ Z - \widehat\E(Z \mid A,Y,R^\dagger=1) \} + \widehat\E(Z \mid A,Y,R^\dagger=1)]\{ A - \Pn(A) \} ) }
\end{equation}
where $\widehat\E(R^\dagger \mid A,Y)$ and $\widehat\E(Z \mid A,Y,R^\dagger=1)$ are estimates of the  missingness propensity score and instrument regression, respectively. For discrete $Y$, these nuisance functions can be based on saturated models, in which case a standard analysis shows that the estimator attains the efficiency bound under no conditions as long as $\E(R^\dagger \mid A,Y)$ and its estimator are bounded away from zero. For continuous $Y$, construction of these estimators would require some smoothing, and a sufficient condition for asymptotic efficiency would be that the estimators achieve faster than $n^{-1/4}$ rates (e.g., in $L_2$ norm). The estimator $\widehat\Psi_n$ can be viewed as solving an estimating equation based on an estimated version of the efficient influence function in \eqref{eq:eif_mar} (with $R^\dagger$ replacing $R$). This is a standard way to construct efficient estimators, and the analysis of corresponding asymptotic properties follows from usual estimating equation techniques \autocite{van2000asymptotic,van2003unified}.

\subsection{Explicit form for bias}

In this subsection we explore the form of the bias of a complete-case approach using the indicator $R^\dagger$, in the presence of one-sided noncompliance. 

\begin{lemma}
Under MCAR ($R \ind (Z,A,Y)$) the bias of $\Psi^\dagger$ from \eqref{eq:psi_dag} is given by
$$ \Psi^\dagger(\Pb)  - \Psi^f(\Pb)  =  (\alpha \beta) \frac{ \Pb(R^\dagger=0 \mid Z=1) }{ \E(A \mid Z=1, R^\dagger=1)} $$
for
\begin{align*}
\alpha &= \E(A \mid Z=1,R^\dagger=1) - \E(A \mid Z=1,R^\dagger=0) \\
\beta &= \E(Y \mid Z=1, A=1) - \E(Y \mid Z=1, A=0) - \Psi^f(\Pb)  .
\end{align*}
\end{lemma}

A proof of Lemma 1 is given in the Appendix. This result indicates that as long as there is missingness, i.e., $\Pb(R^\dagger=0 \mid Z=1) >0$, then there are two scenarios in studies with one-sided noncompliance in which it is valid to do a complete-case analysis (using the indicator $R^\dagger$):
\begin{enumerate}
\item Among those encouraged towards treatment ($Z=1$), compliance ($A=1$) rates are the same regardless of whether the instrument is missing or not (i.e., $\alpha=0$).
\item The mean treatment-control difference in outcomes among those encouraged towards treatment ($Z=1$) is equal to the instrumental variable effect $\Psi^f$ (i.e., $\beta=0$). 
\end{enumerate} 

In general these two no-bias conditions would not be expected to hold. The first condition (involving $\alpha$) cannot be tested since $Z$ is unobserved when $R^\dagger=0$. However the second condition could be assessed (assuming MAR) by estimating $\Psi^f$ under MAR and comparing to the mean treatment-control difference among those encouraged towards treatment with $R^\dagger=1$. Alternatively, to assess bias apart from these conditions, one could simply compare estimates under MCAR and MAR directly.

\section{Second paradox}

In this section we present our second apparent paradox: that, although knowing noncompliance is one-sided does not yield efficiency gains with full non-missing data, such knowledge is in fact informative in settings with missing data. In particular we characterize the change in the efficiency bound that results from exploiting knowledge of one-sided compliance.

\subsection{Efficiency bound without missingness}

Recall from Section 3.1 that the full data efficient influence function is $D(\bO)$ from \eqref{eq:eif_full}, which is also the influence function for the usual IV estimator $\Psi^f(\Pn)$ from \eqref{eq:psihat}. This can be shown to yield the full data efficiency bound
$$ \var\{ D(\bO) \} = \frac{ \sum_z \var(Y - \Psi^f A \mid Z=z)/ \Pb(Z=z) }{ \{\E(A \mid Z=1) - \E(A \mid Z=0)\}^2 } . $$

Now suppose a known compliance rate $\E(A \mid Z=0)=\rho$  among the unencouraged, for some known value $0 \leq \rho < 1-\epsilon$. For example $\rho=0$ corresponds to the one-sided noncompliance scenario in which $Z=0$ implies $A=0$. Then we can define the corresponding estimand
$$ \Psi^f_\rho(\Pb) = \frac{\E(Y \mid Z=1)-\E(Y \mid Z=0)}{\E(A \mid Z=1) - \rho} $$
with efficient estimator 
$$ \Psi^f_\rho(\Pn) = \frac{\Pn(Y \mid Z=1)-\Pn(Y \mid Z=0)}{\Pn(A \mid Z=1) - \rho} = \left[ \frac{1}{\Psi^f(\Pn)} + \frac{\Pn\{(A-\rho)(1-Z) \Pn(Z) \}}{\Pn(ZY) - \Pn(Z) \Pn(Y} \right]^{-1} . $$
This estimator $\Psi^f_\rho(\Pn)$ can be shown to solve the equation
$$ \Pn \Big( (Z-\widehat\pi_n) \Big[Y - \psi) \{ZA + (1-Z)\rho\} \Big] \Big) = 0 $$
in $\psi$ for $\widehat\pi_n=\Pn(Z)$, rather than the equation $\Pn\{(Z-\widehat\pi_n)(Y - \psi A)\}=0$ solved by the standard estimator $\Psi^f(\Pn)$. Although not the case for general $\rho>0$, when $\rho=0$ exactly we have $\Psi^f_\rho(\Pn) = \Psi^f(\Pn)$, so that there is no efficiency gain from incorporating knowledge of one-sided noncompliance. This follows since $ZA+(1-Z)\rho=ZA=A$ when $\rho=0$, so that incorporating the knowledge that $\rho=0$ yields the exact same estimator as if the knowledge was not used. This is to be expected, since if $\E(A \mid Z=0)=0$ exactly then the estimator $\Pn(A \mid Z=0)=\Pn\{A(1-Z)\} / \Pn(1-Z)$ will equal the true parameter $\rho=0$ with probability one (of course the converse is not true: if $\Pn(A \mid Z=0)=0$ we cannot be certain that $\rho=0$).

Therefore, when there is no missing data, incorporating knowledge of one-sided noncompliance (or, in other words, $\rho=0$) simply yields the same standard estimator that does not incorporate this knowledge, and hence does not give any efficiency gains. However, in the next subsection we will show that the story is different when there is missing data (and MCAR holds for the original missingness indicator $R$): then using such knowledge can in fact provide efficiency gains. 

\subsection{Main result}

Here we show our second apparently paradoxical result: that knowledge of one-sided noncompliance gives opportunities for efficiency gains when instrument values are missing, contrary to the setting without such missing data. This follows because the knowledge that $A=1$ implies $Z=1$ allows us to fill in some missing data (among those with $A=1$).

\begin{theorem}
Assume the MCAR condition $R \ind (Z,A,Y)$. Then the efficiency bound $\var\{D^*(\bO)\}$ for estimating $\Psi^f$ is no less than the bound $\var\{ D_{os}(\bO)\}$ when it is also known that noncompliance is one-sided. In particular the difference between the bounds is
$$ \var\{ D^*(\bO)\} - \var\{ D_{os}(\bO)\} = \text{odds}(R=0) \E\Big[ \var\{ D(\bO) \mid A, Y \} \Bigm| A=1 \Big] \Pb(A=1)  \geq 0 . $$
\end{theorem}

\begin{proof}

As noted in Remark 2 we have that $R \ind (Z,A,Y)$ implies $R^\dagger \ind (Z,Y) \mid A$. Therefore, as in Section 3.3  the efficiency bounds are equivalent under $R^\dagger \ind (Z,Y) \mid A$ and $R^\dagger \ind Z \mid (A,Y)$ \autocite{robins1995semiparametric}. This implies the efficient influence function under $R \ind (Z,A,Y)$ with known one-sided noncompliance is given by
\begin{align*}
D_{os}(\bO) &= \frac{R^\dagger}{\E(R^\dagger \mid A,Y)} \Big[ D(\bO) - \E\{ D(\bO) \mid A, Y, R^\dagger=1 \}  \Big] + \E\{ D(\bO) \mid A, Y, R^\dagger=1\} \\
&= \frac{R^\dagger}{\E(R^\dagger \mid A)} \Big[ D(\bO) - \E\{ D(\bO) \mid A, Y \}  \Big] + \E\{ D(\bO) \mid A, Y \}
\end{align*}
where we note that $\E\{D(\bO) \mid A, Y, R^\dagger = 1\} = \E\{D(\bO) \mid A, Y \}$ follows from the fact that $R^\dagger \ind Z \mid (A,Y)$, which is implied by $R^\dagger \ind (Z,Y) \mid A$. Therefore the variance of the efficient influence function (which by definition is the efficiency bound) is 
\begin{align*}
\var\{ D_{os}(\bO)\} 
&= \E\left( \frac{R^\dagger}{\E(R^\dagger \mid A)^2} \Big[ D(\bO) - \E\{ D(\bO) \mid A, Y \}  \Big]^2 + \E\{ D(\bO) \mid A, Y\}^2 \right) \\
&=  \E\left[ \frac{\var\{ D(\bO) \mid A, Y \}}{\E(R^\dagger \mid A)}  \right] + \var\Big[ \E\{ D(\bO) \mid A, Y\}\Big]  \\
&= \var\{D(\bO)\} +  \E\left[ \frac{\Pb(R^\dagger=0 \mid A)}{\Pb(R^\dagger=1 \mid A)} \var\{ D(\bO) \mid A, Y \} \right] \\
&= \var\{D(\bO)\} + \text{odds}(R=0) \E\Big[ (1-A) \var\{ D(\bO) \mid A, Y \} \Big] . 
\end{align*}
The first equality follows since the cross term in the variance is zero by iterated expectation, the second by iterated expectation and since $\E[\E\{D(\bO) \mid A,Y\}]=0$, the third using the law of total variance, and the fourth using the definition of $R^\dagger$. Rearranging the last line yields the result.
\end{proof}

Theorem 2 shows that, in general, efficiency gains are possible when $Z$ is partially missing by incorporating knowledge of one-sided noncompliance. The only ways the bounds $\var\{ D^*(\bO)\}$ and $\var\{ D_{os}(\bO)\}$ can be equal are in the two unusual no-variance situations mentioned in Section 3.3, or if no one is treated (i.e., there are only never-takers). This is different from the setting without missing data, where such efficiency gains are not possible; there knowledge of one-sided noncompliance leads to the exact same estimators that would be used without such knowledge. The intuition behind the gain in efficiency comes from the fact that knowledge of one-sided noncompliance allows us to fill in some missing data; in particular we know that those who are treated must have had missing instruments equal to $Z=1$, since those with $Z=0$ cannot access treatment. This also explains why the increase in efficiency is proportional to the fraction treated $\Pb(A=1)$.

\section{Discussion}

In this paper we discussed two paradoxes related to bias and efficiency in instrument variable studies, in the common setting where instrument values are partially missing, and noncompliance is one-sided. Our first paradox is that complete-case analyses are biased, even when values are missing completely at random, if knowledge of one-sided noncompliance is taken into account. This is because incorporating one-sided noncompliance information (by filling in $Z=1$ whenever $A=1$) induces a dependence between missingness and treatment, invalidating MCAR and thus a complete-case analysis. Our second paradox is that incorporating information about one-sided noncompliance generally also leads to efficiency gains (i.e., the efficiency bound changes) when there is missing data, but not when no data are missing. This second result is due to the fact that the one-sided noncompliance allows us to fill in some missing values (again since we know $Z=1$ whenever $A=1$), thus alleviating some of the information loss due to missing instrument values. 

We hope our work might spur more research on missingness in instrumental variable designs, since we have shown that some interesting and unexpected phenomena can arise. More generally, in the same spirit as \textcite{van2003unified}, we also hope our paper becomes part of a larger stream of work that gives a more connected and simultaneous treatment of causal inference and missing data problems, since issues related to both often occur together in practice.

\stepcounter{section}
\printbibliography[title={\thesection \ \ \ References}]

\section{Appendix}

\subsection{Proof of Lemma 1}

First note that the true IV estimand simplifies to
$$ \Psi^f =  \frac{\E(Y \mid Z=1)-\E(Y \mid Z=0)}{\E(A \mid Z=1) } $$
since $Z=0$ implies $A=0$. This together with $R^\dagger \ind (Z,Y) \mid A$ also implies
$$ \E(Y \mid Z=0, R^\dagger=1) = \sum_a \E(Y \mid Z=0, A=a) \Pb(A=a \mid Z=0, R^\dagger=1)  = \E(Y \mid Z=0) , $$
so the complete-case estimand based on $R^\dagger$ is similarly given by
$$ \Psi^\dagger = \frac{\E(Y \mid Z=1,R^\dagger=1)-\E(Y \mid Z=0)}{\E(A \mid Z=1,R^\dagger=1) } . $$

Therefore the bias of $\Psi^\dagger$ is 
\begin{align*}
\Psi^\dagger - \Psi^f &= \frac{\E(Y \mid Z=1,R^\dagger=1)-\E(Y \mid Z=0)}{\E(A \mid Z=1,R^\dagger=1) } - \frac{\E(Y \mid Z=1)-\E(Y \mid Z=0)}{\E(A \mid Z=1) }  \\
&= \frac{\{ \E(Y \mid Z=1,R^\dagger=1)-\E(Y \mid Z=0) \} \E(A \mid Z=1) }{\E(A \mid Z=1,R^\dagger=1) \E(A \mid Z=1) }  \\
& \hspace{.5in} - \frac{ \{ \E(Y \mid Z=1)-\E(Y \mid Z=0) \} \E(A \mid Z=1, R^\dagger=1) }{ \E(A \mid Z=1,R^\dagger=1) \E(A \mid Z=1) }  \\
&= \frac{ \E(Y \mid Z=1, R^\dagger=1) - \E(Y \mid Z=1) }{\E(A \mid Z=1, R^\dagger=1)} + \Psi^f \left\{ \frac{ \E(A \mid Z=1)}{ \E(A \mid Z=1, R^\dagger=1)} - 1 \right\} \\
&= \Big\{ \E(Y \mid Z=1, A=1) - \E(Y \mid Z=1, A=0) \Big\} \frac{ \E(A \mid Z=1,R^\dagger=1) - \E(A \mid Z=1) }{ \E(A \mid Z=1, R^\dagger=1)}  \\
& \hspace{.5in} + \Psi^f \left\{ \frac{ \E(A \mid Z=1)}{ \E(A \mid Z=1, R^\dagger=1)} - 1 \right\} \\
&= \Big\{ \E(Y \mid Z=1, A=1) - \E(Y \mid Z=1, A=0) - \Psi^f \Big\} \\
& \hspace{.5in} \times  \frac{ \E(A \mid Z=1,R^\dagger=1) - \E(A \mid Z=1) }{ \E(A \mid Z=1, R^\dagger=1)} \\
&= \Big\{ \E(Y \mid Z=1, A=1) - \E(Y \mid Z=1, A=0) - \Psi^f \Big\}  \\
& \hspace{.5in} \times \Big\{ \E(A \mid Z=1,R^\dagger=1) - \E(A \mid Z=1,R^\dagger=0) \Big\} \times \left\{ \frac{ \Pb(R^\dagger=0 \mid Z=1) }{ \E(A \mid Z=1, R^\dagger=1)} \right\} 
\end{align*}
where the first equality follows by expressions for $\Psi^f$ and $\Psi^\dagger$ above, the second by rearranging, the third by subtracting and adding $\E(Y \mid Z=1)$ in the numerator of the first term, the fourth since
\begin{align*}
\E(Y &\mid Z=1, R^\dagger=1) - \E(Y \mid Z=1) \\
&= \Big\{ \E(Y \mid Z=1, A=1) - \E(Y \mid Z=1, A=0) \Big\} \E(A \mid Z=1,R^\dagger=1) \\
& \hspace{.4in} - \Big\{ \E(Y \mid Z=1, A=1) - \E(Y \mid Z=1, A=0) \Big\} \E(A \mid Z=1) \\
& \hspace{.4in} + \E(Y \mid Z=1, A=0)  - \E(Y \mid Z=1, A=0) \\
&= \Big\{ \E(Y \mid Z=1, A=1) - \E(Y \mid Z=1, A=0) \Big\}  \Big\{ \E(A \mid Z=1, R^\dagger=1) - \E(A \mid Z=1) \Big\} ,
\end{align*}
the fifth by rearranging,  and the sixth since
\begin{align*}
\E(A \mid Z=1) &= \Big\{ \E(A \mid Z=1, R^\dagger=1) - \E(A \mid Z=1, R^\dagger=0) \Big\} \E(R^\dagger \mid Z=1) \\
& \hspace{.4in} + \E(A \mid Z=1, R^\dagger=0) 
\end{align*}
and rearranging. This yields the desired result.

\end{document}